\documentclass[11pt]{article}

  \usepackage{amssymb,amsfonts,amsmath,mathrsfs,float, fullpage}
  \usepackage{hyperref}
  \usepackage[usenames,dvipsnames]{color}
  \usepackage{tikz}
   \usetikzlibrary{calc}
   
  \usepackage[T1]{fontenc}
  \usepackage{tgpagella}
  \usepackage{mathpazo}
  \usepackage[heavycircles]{stmaryrd}
  \newtheorem{definition}{Definition}
  \newtheorem{theorem}{Theorem}
  \newtheorem{lemma}[theorem]{Lemma}
  \newtheorem{corollary}[theorem]{Corollary}

  \newtheorem{proposition}[theorem]{Proposition}

  \newenvironment{proof}[1][]{
   \begin{trivlist}
    \item[\hspace{\labelsep}{\sc\noindent Proof#1:\/}]}
    {{\hfill$\Box$}
   \end{trivlist}}

  \newcommand{\R}{\mathbb{R}} 
  \newcommand{\C}{\mathbb{C}} 
  \newcommand{\F}{\mathbb{F}} 
  \newcommand{\N}{\mathbb{N}} 
  \newcommand{\pmset}[1]{\{-1,1\}^{#1}} 
  \newcommand{\bset}[1]{\{0,1\}^{#1}} 
  \DeclareMathOperator{\aut}{Aut} 
  \DeclareMathOperator{\stab}{Stab} 
  \DeclareMathOperator{\orb}{Orb} 
  
  \newcommand{\st}{:\,} 
  \DeclareMathOperator{\rank}{rank}
  
  \DeclareMathOperator{\Tr}{Tr} 
  \newcommand{\eps}{\varepsilon}
  \newcommand{\HS}{\mathcal{H}} 

  \newcommand{\beq}{\begin{equation}} 
  \newcommand{\eeq}{\end{equation}}
  \newcommand{\beqn}{\begin{equation*}}
  \newcommand{\eeqn}{\end{equation*}}
  \newcommand{\beqr}{\begin{eqnarray}}
  \newcommand{\eeqr}{\end{eqnarray}}
  \newcommand{\beqrn}{\begin{eqnarray*}}
  \newcommand{\eeqrn}{\end{eqnarray*}}
  \newcommand{\bmline}{\begin{multline}}
  \newcommand{\emline}{\end{multline}}
  \newcommand{\bmlinen}{\begin{multline*}}
  \newcommand{\emlinen}{\end{multline*}}

  \newcommand{\ie}{{i.e.}}

\begin{document}

\title{\bf \huge Violating the Shannon capacity of metric graphs with entanglement}

\author{Jop Bri\"{e}t\footnote{CWI, Amsterdam. Supported by the EU 7th framework grant QCS. Email:j.briet@cwi.nl} \and Harry Buhrman\footnote{CWI and University of Amsterdam. Supported by the EU 7th framework grant QCS} \and Dion Gijswijt\footnote{CWI and TU Delft}}

\date{}
\maketitle

\begin{abstract}
The Shannon capacity of a graph $G$ is the maximum asymptotic rate at which messages can be sent with zero probability of error through a noisy channel with confusability graph~$G$.
This extensively studied graph parameter disregards the fact that on atomic scales, Nature behaves in line with quantum mechanics.
Entanglement, arguably the most counterintuitive feature of the theory, turns out to be a useful resource for communication across noisy channels.
Recently, Leung, Man\v{c}inska, Matthews, Ozols and Roy~[Comm.\ Math.\ Phys.\ 311, 2012]
presented two examples of graphs whose Shannon capacity is strictly less than the capacity attainable if the sender and receiver have entangled quantum systems.
Here we give new, possibly infinite, families of graphs for which the entangled capacity exceeds the Shannon capacity.
\end{abstract}

\section{Introduction}

A sender transmits a message to a receiver. 
The main problem that information theory addresses is that {\em noise} could make the sender's announcement ambiguous.
To  analyze this problem, one models a noisy communication channel by an input alphabet~$\mathcal S$, an output alphabet~$\mathcal R$ and a set of conditional probabilities $P(a|x)$ for the probability that the receiver gets the letter~$a\in\mathcal R$ when the sender transmits the letter~$x\in\mathcal S$.
Two input letters $x$ and $y$ can then be confused with one another if they can lead to the same signal on the receiver's end of the channel, that is, if there is an output letter $a$ such that both the probability $P(a|x)$ of the receiver getting $a$ when the sender sent $x$, and the probability $P(a|y)$ of the receiver getting $a$ when the sender sent~$y$, are nonzero.
To cope with noise, the communicating parties could agree that the sender uses only input letters that lead to distinct signals on the receiver's end, in which case the sender restricts to some set $T\subseteq \mathcal S$ such that for every pair of distinct inputs $x, y\in T$ and $a\in \mathcal R$, at least one of the probabilities $P(a|x)$ and $P(a|y)$ vanishes. 
In a celebrated paper, Shannon~\cite{Shannon:1956} initiated the study of the {\em zero error capacity}, the maximum rate of error-free communication with sequential uses of a memoryless channel. A channel is memoryless if using it does not change its behavior on later messages.
By encoding messages into {\em words} consisting of multiple input symbols that are transmitted in sequence, a channel can sometimes be used more efficiently than if non-confusable symbols are simply concatenated.
Shannon demonstrated this with a famous example of a channel with five inputs and outputs.
At most two symbols can be sent perfectly with one use of the channel.
Instead of the expected four, {\em five} messages can be sent perfectly with two uses of the channel (see Figure~\ref{fig:cfive}, (i)).

\begin{figure}[t] 
\begin{center}
 \begin{minipage}{.2\textwidth} 
\begin{center}
 \includegraphics{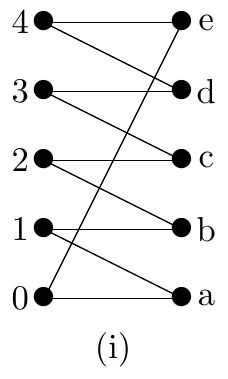}
\end{center}
 \end{minipage} 
 \begin{minipage}{.2\textwidth} 
\begin{center}
 \includegraphics{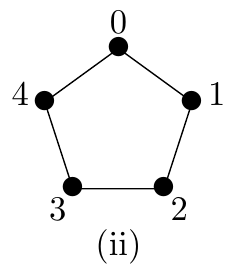}
\end{center}
 \end{minipage} 
\caption{(i): A channel with five inputs (numbers) and five outputs (letters), where an input symbol $x$ is connected to an output symbol $a$ if $P(a|x)>0$. 
Notice that no two of the five {\em pairs} $(0,2), (1,4), (2,1), (3,3), (4,0)$ can be confused with one another, as either the first or the second two symbols are non-confusable.
(ii): the confusability graph of the channel, the five cycle~$C_5$.}\label{fig:cfive}
 \end{center}
\end{figure}

Associated with a noisy channel is its {\em confusability graph} $G = (V,E)$, with as vertex set $V$ the input alphabet of the channel and edge set $E$ consisting of those pairs of inputs that can be confused (see Figure~\ref{fig:cfive}, (ii)).
The largest number of messages that can be sent through a channel with zero probability of error equals $\alpha(G)$, the maximum cardinality of an independent set in the graph~$G$.
The graph $G^{\boxtimes n}$ has as vertex set $V^n$, the set of all $n$-tuples of elements from $V$. Two different vertices $(u_1,\dots,u_n)$ and  $(v_1,\dots,v_n)$ are adjacent in $G^{\boxtimes n}$ if and only if they can be confused (\ie, if for every $i = 1,\dots,n$, either $u_i = v_i$ or $u_i$ is adjacent to $v_i$ in $G$).
The independence number~$\alpha(G^{\boxtimes n})$ thus gives the maximum number of pairwise non-confusable messages that can be sent using $n$-letter codewords.
The {\em Shannon capacity} of the confusability graph,
\beqn
\Theta(G) = \sup_{n}\sqrt[n]{\alpha(G^{\boxtimes n})} = \lim_{n\to\infty}\sqrt[n]{\alpha(G^{\boxtimes n})},
\eeqn
gives the zero-error capacity of a channel.
\medskip

Inevitably, devices used for information processing are subject to the laws of physics, and on atomic scales {\em quantum mechanics} is currently the most accurate model of Nature.
Thirty years before Shannon's paper was published, Einstein, Podolsky and Rosen~\cite{Einstein:1935} pointed out an anomaly of quantum mechanics that allows spatially separated parties to establish peculiar correlations: {\em entanglement}. 
Later, Bell~\cite{Bell:1964} proved that local measurements on a pair of spatially separated, entangled quantum systems, can give rise to joint probability distributions that {\em violate} certain inequalities (now called Bell inequalities) satisfied by any distribution that may arise in classical mechanics.
Experimental results of Aspect et al.\ \cite{Aspect:1981} give strong evidence that Nature indeed allows distant physical systems to be correlated in such non-classical ways.
Motivated by these important discoveries, one defines the {\em zero-error entanglement-assisted capacity} of a classical channel to be the maximum rate at which messages can be transmitted when the sender and receiver share a pair of entangled quantum systems (the precise model is described below).

Analogous to the classical setting, Cubitt, Leung, Matthews and Winter~\cite{Cubitt:2010} 
defined the graph parameter $\alpha_q(G)$ (a quantum variant of the independence number), and proved that it equals the maximum number of pairwise non-confusable messages that can be sent with a single use of  a noisy channel and shared entanglement.
They found examples of graphs for which $\alpha_q(G)>\alpha(G)$, showing that the use of entanglement can increase the ``one-shot'' zero-error capacity of a channel (more examples where found recently by Man\v{c}inska, Scarpa and Severini~\cite{Mancinska:2012}). 
This result was seen as a surprise in light of the fact that entanglement {\em cannot} increase the transmission rate if one only demands that the {\em probability} of confusion goes to zero with the number of uses of the channel, as was proved by Bennett, Shor, Smolin and Thapliyal~\cite{Bennett:2002}.
Of course, the next question was if the {\em entangled capacity}, defined by
\beqn
\Theta_q(G) = \lim_{n\to\infty}\sqrt[n]{\alpha_q(G^{\boxtimes n})},
\eeqn
could be strictly greater than the Shannon capacity.

In contrast to the combinatorial nature of the Shannon capacity, the entangled capacity can sometimes be lower bounded using geometric constructions.
An {\em orthonormal representation} of a graph is a map that sends each of the vertices of the graph to a vector on a Euclidean unit sphere, such that adjacent vertices are sent to orthogonal vectors.\footnote{We stress that in our definition {\em orthogonality corresponds to adjacency}. Some authors prefer to demand orthogonality for non-adjacent vertices instead.}
It turns out that if a graph~$G$ has a $d$-dimensional orthonormal representation, then $\alpha_q(G)$, and therefore $\Theta_q(G)$, is at least the number of disjoint~$d$-cliques in~$G$ (see Proposition~\ref{prop:leungbound}).

Using the above argument, Leung, Man\v{c}inska, Matthews, Ozols and Roy~\cite{Leung:2012}  recently found the first two examples of graphs whose entangled capacity exceeds the Shannon capacity. Their graphs are based on the exceptional root systems $E_7$ and $E_8$, giving a graph $G_{E_7}$ on 63 vertices such that $\Theta_q(G_{E_7})/\Theta(G_{E_7}) = 9/7$ and a graph $G_{E_8}$ on 157 vertices such that $\Theta_q(G_{E_8})/\Theta(G_{E_8}) = 15/9$.
They also show that their construction fails on any of the (four) infinite families of root systems, in the sense that it yields graphs whose Shannon and entangled capacities are equal.

The above-described lower-bound technique perhaps makes the {\em orthogonality graph}---whose vertices are the binary strings of length~$n$, and whose edges are all pairs with Hamming distance~$n/2$ (for $n$ even)---the most natural candidate to separate the two capacities.
This graph lies at the heart of many constructions that show a separation between some classical quantity and its quantum analogue, such as in communication complexity~\cite{Buhrman:1998, Brassard:1999}, in Bell-inequality violations~\cite{Avis:2006, Cameron:2007, Godsil:2008} and in the ``one-shot'' zero-error capacity of a noisy channel (where one compares $\alpha(G)$ to $\alpha_q(G)$).

By using the $\pmset{}$-basis for bits, one directly obtains an $n$-dimensional orthonormal representation for the orthogonality graph, since two strings have Hamming distance $n/2$ (meaning they are adjacent) if and only if they are orthogonal.
A {\em Hadamard matrix} is a square matrix with entries in $\{-1,1\}$ such that its rows are mutually orthogonal.
If a Hadamard matrix of size $n$ exists, then its rows thus give a clique of size~$n$ in the orthogonality graph.
The fact that this graph is vertex transitive then implies that it in fact has at least $2^n/n^2$ disjoint cliques (see Lemma~\ref{lem:transcliques} and Proposition~\ref{prop:Htrans} below), giving the same lower bound on its entangled capacity.
It is well known that Hadamard matrices exist when $n$ is a power of $2$ and the famous Hadamard conjecture states that they exist whenever $n$ is a multiple of~$4$.
Although this conjecture remains unproven, it is widely believed to be true.

An indication that the orthogonality graph might exhibit a separation between the Shannon and entangled capacities is given by a deep result of Frankl and R\"{o}dl~\cite{Frankl:1987} showing that if $n$ is a large enough multiple of~$4$, then the independence number is less than $(2-\eps)^n$ for some $\eps>0$ independent of~$n$.
But despite effort from the quantum-information community, it remains unknown if this graph gives such a separation.
Our main result shows that under certain conditions such a result does hold for a {\em ``quarter''} of the orthogonality graph.

%

%
\section{Our results}

In this paper we present two new, possibly infinite, families of basic graphs whose entangled capacity exceeds the Shannon capacity.
The graphs are defined as follows. Let~$H_n$ be the graph with as vertex set all binary strings of odd length $n$ and even Hamming weight, and as edge set the pairs with Hamming distance~$(n+1)/2$. Let $G_n$ be the subgraph of~$H_n$ induced by the strings of Hamming weight $(n+1)/2$.
We prove the following.

\begin{theorem}\label{thm:capacity1}
Let $p$ be an odd prime such that there exists a Hadamard matrix of size $4p$.
Then, for $n = 4p-1$ and $G$ either $G_n$ or $H_n$, we have 
\beqn
\frac{\Theta_q(G)}{\Theta(G)} \geq\Omega\Big(\frac{2^{0.752 p}}{p^{5/2}}\Big).
\eeqn
\end{theorem}

Notice that the graph $H_n$ is a subgraph of the orthogonality graph on~$\bset{n+1}$, induced by the $(n+1)$-bit strings with even Hamming weight and first coordinate equal to~$0$.
Based on constructions of Hadamard matrices due to Scarpis~\cite{Scarpis:1898} and Paley~\cite{Paley:1933}, Theorem~\ref{thm:capacity1} holds for any prime $p$ such that $4p-1 = q^k$ for some odd prime $q$ and positive integer~$k$. 
The first three examples of such $(p,q)$ pairs for $k = 1$  are $(3, 11), (5, 19)$ and $(11, 43)$.
Examples for $p\approx 10^{12}$ and $k=1$ can readily  be generated with little computing power.
The subset of strings in the vertex set of $G_n$ that have zeroes on the last $(n-7)/4$ coordinates is an independent set of size $\Omega\big(2^{0.29 n})$, showing that the Shannon capacity of $G_n$ is  exponential in $n$.
Since $G_n$ is an induced subgraph of $H_n$, the same holds for the latter graph.

We observe that the results of~\cite{Leung:2012} imply that the sequence of graphs $(G_{E_8}^{\boxtimes n})_{n\in \N}$ has a capacity ratio $\Theta_q/\Theta$ that grows as roughly $|V(G_{E_8}^{\boxtimes n})|^{0.101}$, where $V(G)$ is the number of vertices of the graph $G$ (the graph $G_{E_8}$ gives better dependence on the number of vertices than $G_{E_7}$ does). 
Our results show that the family of graphs $G_n$ gives a slightly higher ratio of roughly $|V(G_n)|^{0.187}$.

Theorem~\ref{thm:capacity1} follows directly from the following lemmas, which give lower and upper bounds on the entangled capacity and Shannon capacity, respectively.

\begin{lemma}\label{lem:qcapacity}
Let $n$ be a positive integer such that there exists a Hadamard matrix of size $(n+1)$.
Then, for $G$ either~$G_n$ or $H_n$, we have $\Theta_q(G) \geq |V(G)|/(n+1)^2$.
\end{lemma}

\begin{lemma}\label{lem:shannonbound}
Let $p$ be an odd prime and let $n = 4p-1$. Then, for $G$ either $G_n$ or $H_n$, we have 
\beqn
\Theta(G) \leq {n\choose 0} + {n\choose 1} + \cdots + {n\choose p-1}.
\eeqn
\end{lemma}

Aside from relying on a few basic facts of graph theory and the theory of finite fields, our proofs of these lemmas are straightforward and self-contained.
To obtain the asymptotic bound of Theorem~\ref{thm:capacity1} we upper bound the binomial sum of Lemma~\ref{lem:shannonbound} by the well-known estimate $2^{n H(p/n)}$,
where $H(t) = -t\log_2 t - (1-t)\log_2 (1-t)$ is the binary entropy function, and use the bound $|V(G_n)| \geq \Omega(2^n/\sqrt{n})$.

\section{The entangled capacity of a graph}

In this section we give the formal definition of the entangled capacity of a graph.
Let $G$ be a finite simple undirected graph, with vertex set $V(G)$ end edge set $E(G)\subseteq  {V\choose 2}$.
The {\em strong graph product} $G\boxtimes H$ of two graphs $G$ and $H$ has as vertex set all pairs $(u,v)\in V(G)\times V(H)$. Two vertices $(u,v)$ and $(u',v')$ are adjacent in $G\boxtimes H$ if and only if $u$ and $u'$ are adjacent in $G$ or equal, and $v$ and $v'$ are adjacent in $H$ or equal.
For example, if $\{u,u'\}\in E(G)$ and $\{v,v'\}\in E(H)$, then the four pairs $(u,v)$, $(u,v')$, $(u',v)$ and $(u',v')$ form a $4$-clique in $G\boxtimes H$.\footnote{Hence the symbol~$\boxtimes$.} 
We denote by $G^{\boxtimes n}$ the $n$-fold strong graph product of $G$ with itself.
Recall that the {\em Shannon capacity} of  $G$ is defined as $\Theta(G) = \sup_{n}\big(\alpha(G^{\boxtimes n})\big)^{1/n}$.

The {\em entangled capacity} of a graph is defined as follows.

\begin{definition}[Entangled capacity of a graph]\label{def:qcap}
Let $G = (V,E)$ be a graph.
Define $\alpha_q(G)$ to be the largest natural number $M$ such that there exists a Hilbert space $\HS$, a trace-1 positive semidefinite operator~$\rho$ on $\HS$ and a  positive semidefinite operator $\rho(u)_i$ on~$\HS$ for every $i \in\{1,\dots,M\}$ and $u\in V$  satisfying 
\begin{enumerate}
\item $\sum_{u\in V} \rho(u)_i = \rho$ for every $i\in\{1,\dots,M\}$ \label{eq:qindep-sumrho}\\[.1cm]
\item $\rho(u)_i\rho(u)_j = 0$ for every $u\in V$ and $i\ne j$\\[.1cm]
\item $\rho(u)_i\rho(v)_j = 0$ if $\{u,v\}\in E$ and $i\ne j$.\label{eq:qindep-trzero}
\end{enumerate}
The {\em entangled capacity} of $G$ is defined by 
\beqn
\Theta_q(G) = \sup_{n}\big(\alpha_q(G^{\boxtimes n})\big)^{1/n}.
\eeqn
\end{definition}

The parameter $\alpha_q(G)$ satisfies $\alpha_q(G) \geq \alpha(G)$ and is a generalization of the independence number. To see this, restrict in Definition~\ref{def:qcap} the space $\HS$ to be one dimensional and add the further restrictions $\rho = 1$ and $\rho(u)_i \in\{0,1\}$.
Say that a vertex $u\in V$ {\em gets label~$i$} if $\rho(u)_i = 1$.
Condition~1 says that exactly one vertex gets label~$i$, Condition~2 says that each vertex gets at most one label and Condition~3 says that no two adjacent vertices belong to the privileged subset of labeled vertices.
Hence, the system of numbers $\rho(u)_i$ gives an independent set of size $M$, namely the set $\{u\st \rho(u)_i =1 \text{ for some $i$}\}$.
Since $\alpha_q(G)$ relaxes this characterization of $\alpha(G)$, it follows that $\alpha_q(G) \geq \alpha(G)$.

By using tensor products of the operators $\rho$ and $\rho(u)_i$ it is not hard to see that $\alpha_q$ is super-multiplicative under strong graph powers. It follows that $\Theta_q(G) \geq \alpha_q(G)$ and (by Fekete's lemma) that $\Theta_q(G) = \lim_{n\to\infty}\big(\alpha_q(G^{\boxtimes n})\big)^{1/n}$.

\section{Entanglement-assisted communication}

In this section we describe the model of zero-error entanglement-assisted communication over classical channels. 
Readers who are  familiar with this model or want to move on to the proof of the main result can safely skip this section.
We start with some basic definitions of quantum information theory.
For more details we refer to Nielsen and Chuang~\cite{Nielsen:2000}.

\paragraph{Shared entangled states.}
A {\em state} is a positive-semidefinite matrix whose trace equals~1.
We identify a matrix of size $d\!\times\!d$ with a linear operator on~$\C^d$ in the obvious way.
A state should be though of as describing the configuration of a {\em quantum system}: an abstract physical object, or collection of objects, on which one can perform experiments.
Associated with a quantum system~$\mathsf Q$ is a complex Euclidean vector space~$\mathcal Q = \C^d$, for some dimension~$d$.
The possible configurations of $\mathsf Q$ are the states on~$\mathcal Q$.

Suppose the sender and receiver hold quantum systems $\mathsf S$ and $\mathsf R$, respectively.
Associated with the sender's system is a  space $\mathcal X = \C^n$, and associated with the receiver's system is a space $\mathcal Y= \C^m$.
Then, by definition, the possible configurations of the {\em joint} system $(\mathsf S,\mathsf R)$ are the states on~$\mathcal X\otimes\mathcal Y$.
If the system $(\mathsf S,\mathsf R)$ is in the state $\rho$, then the sender and receiver are said to {\em share the state~$\rho$}.
A state on $\mathcal X\otimes\mathcal Y$ is {\em entangled} if it is not a convex combination of states of the form $\rho_S\otimes\rho_R$, where $\rho_S$ is a state on~$\mathcal X$ and $\rho_R$ a state on~$\mathcal Y$.

\paragraph{Measurements.}
Let $\mathcal S$ be a finite set and let $\mathsf Q$ be a quantum system with associated vector space $\mathcal Q = \C^d$.
A {\em measurement} on the system $\mathsf Q$ with {\em outcomes} in $\mathcal S$ is a  system of positive semidefinite matrices $M^s$ on~$\mathcal{Q}$, $s\in \mathcal S$, that satisfies 
\beqn
\sum_{s\in \mathcal S} M^s=I_\mathcal Q,
\eeqn
where $I_\mathcal Q$ denotes the identity on $\mathcal Q$.

Let $\{A^s\in\C^{n\times n}\st s\in\mathcal S\}$ be a measurement on the sender's quantum system~$\mathsf S$.
The numbers 
$$
p_s = \Tr\big((A^s\otimes I_\mathcal Y) \rho\big)
$$
define a probability distribution on~$\mathcal S$. 
This follows easily from the properties of the matrices $A^s$ and $\rho$ and the fact that for positive semidefinite matrices $A$ and $B$, we have $\Tr(AB)\geq 0$.

The {\em partial trace function} over $\mathcal X$ of a matrix $M$ on $\mathcal X\otimes \mathcal Y$ is defined by 
\beqn
\Tr_\mathcal X(M) = (e_1\otimes I_\mathcal Y)^{\mathsf T} M (e_1\otimes I_\mathcal Y) + \cdots + (e_n\otimes I_\mathcal Y)^{\mathsf T} M (e_n\otimes I_\mathcal Y),
\eeqn 
where $e_1,\dots,e_n$ are the canonical basis vectors for~$\mathcal X$. 
This function yields an $m\!\times\! m$ matrix (\ie, a linear operator on the space~$\mathcal Y$).
It is not hard to see that each of the matrices 
\beqn
\rho(s) = \frac{\Tr_\mathcal X(A^s\otimes I_\mathcal Y\,\rho) }{p_s}
\eeqn
are in fact states on the space~$\mathcal Y$ associated to the receiver's system~$\mathsf R$.

The postulates of quantum mechanics dictate that if the sender {\em performs the measurement defined by the matrices $A^s$ on her system~$\mathsf S$}, then the following two things happen: 
\begin{enumerate}
\item she obtains outcome $s\in \mathcal S$ with probability~$p_s$,\\[.1cm]
\item the receiver's system~$\mathsf R$ is left in the state $\rho(s)$ on~$\mathcal Y$.
\end{enumerate}

For some finite set $\mathcal R$,  the receiver can perform a measurement $\{B^t\in\C^{m\times m}\st t\in\mathcal R\}$ on $\mathsf R$, and he will obtain outcome~$t$ with probability $\Tr\big(B^t\rho(s)\big)$.
The joint probability of the sender and receiver obtaining outcomes $s$ and $t$, respectively,  is then given by $\Tr\big((A^s\otimes B^t)\rho\big)$.
If the state $\rho$ is {\em not} entangled, then this probability distribution can be described by a classical local-hidden-variable model.
Entanglement is thus necessary to obtain non-classical, quantum distributions.

\paragraph{Entanglement-assisted communication.}

To send messages across a noisy channel defined by input alphabet~$\mathcal S$, output alphabet~$\mathcal R$ and conditional probability distribution~$P$, the sender and receiver can use shared entanglement as follows. 
Let $\rho$ be a state shared between the sender and receiver.
Let $M$ be a positive integer and for every $i\in\{1,\dots,M\}$ let $\{A_i^s\in\C^{n\times n}\st s\in\mathcal S\}$ be a measurement on the sender's system~$\mathsf S$ with outcomes in~$\mathcal S$.
For every $t\in\mathcal R$, let $\{B_t^j\in\C^{m\times m}\st t\in\mathcal R\}$ be a measurement on the receiver's system~$\mathsf R$ with outcomes in $\{1,\dots,M\}$.
Suppose that for every $i\ne j$ and $(s,t)\in\mathcal S\times\mathcal R$ such that $P(t|s)\ne 0$, we have
 \beqn
 \Tr\big((A_i^s\otimes B_t^j)\rho\big) = 0.
 \eeqn
 
To communicate the index~$i$, the sender can then perform the $i$th measurement on her system and sends her outcome $s\in\mathcal \mathcal S$ through the channel. 
The receiver gets a message $t\in\mathcal R$ satisfying $P(t|s)\ne 0$. 
The above discussion shows that if the receiver then performs the measurement labeled by~$t$, he obtains outcome $i$ with probability~1.

The link between this model and Definition~\ref{def:qcap} is given by the following theorem.
Let us denote by $\alpha_q'(\mathcal S,\mathcal R,P)$ the maximum number~$M$ such that a state $\rho$ and matrices $A_i^s, B_t^j$ with the above property exists.

\begin{theorem}[Cubitt, Leung, Matthews and Winter~\cite{Cubitt:2010}]
Let $(\mathcal S,\mathcal R,P)$ be a noisy channel and let $G$ be its confusability graph.
Then, $\alpha_q'(\mathcal S,\mathcal R,P) = \alpha_q(G)$.
\end{theorem}

\section{Preliminaries}
\label{sec:prelims}

\paragraph{Notation.}

We use the following notation.
\begin{itemize}
\item For  strings $x,y\in\bset{n}$, let $d(x,y)$ denote their Hamming distance.
\item For vectors $u,v\in\R^n$, let $u\cdot v$ denote their  Euclidean inner product.
\item For a prime number $p$, we write $\F_p$ for a finite field consisting of $p$ elements. 
\item For vectors $u,v\in\F_p^m$, let $\langle u,v\rangle$ denote their inner product over $\F_p$.
\item For a field~$\F$, we denote by $\F[v_1,\dots,v_n]$ the ring of $n$-variate polynomials with coefficients in~$\F$.
\end{itemize}

\paragraph{Some basic graph theory.}

Let $G$ be a graph.
A permutation of the vertices $\pi: V(G) \to V(G)$ is an {\em automorphism} of $G$ if for every $u,v\in V(G)$, the pair $\{\pi(u),\pi(v)\}$ is an edge if and only if $\{u,v\}$ is an edge.
Let $\aut(G)$ denote the group of automorphisms of $G$.
For $u\in V(G)$, the set
$
\orb(u) = \{\pi(u)\st \pi\in \aut(G)\}
$
is the {\em orbit} of~$u$ and the set
$
\stab(u) = \{\pi\in \aut(G)\st \pi(u) = u\}
$
is the {\em stabilizer} of~$u$.

\begin{definition}
A graph $G$ is {\em vertex transitive} if for every vertex $u\in V(G)$, we have $\orb(u) = V(G)$.
\end{definition}

\begin{lemma}[The Orbit-Stabilizer Theorem~\cite{Alperin:1995}]
 Let $G$ be a graph and $u\in V(G)$. Then 
 $$|\orb(u)|\cdot |\stab(u)| = |\aut(G)|.$$
\end{lemma}

\begin{corollary}\label{cor:utov}
Let $G$ be a vertex transitive graph and $u,v\in V(G)$. Then, there are exactly $|\aut(G)|/|V(G)|$ automorphism of $G$ that map $u$ to $v$.
\end{corollary}

\begin{proof}
Since $G$ is vertex transitive, there exists an automorphism $g\in\aut(G)$ such that $g(u) = v$. Consider the set of automorphisms $g\cdot \stab(u) = \{g h\st h\in\stab(u)\}$.
Clearly $g'(u) = v$ for every $g'\in g\cdot \stab(u)$.
We claim that $g\cdot \stab(u)$ contains all automorphisms that map $u$ to $v$.
To see this, notice that for any $g''\in\aut(G)$ such that $g''(u) = v$, we have $g^{-1}g''\in\stab(u)$ and hence $g'' = g(g^{-1}g'')\in g\cdot \stab(u)$.
Since $gh = gh'$ implies that $h = h'$, we have $|g\cdot \stab(u)| = |\stab(u)|$.
The claim follows because, by the Orbit-Stabilizer Theorem, we have $|\stab(u)| = |\aut(G)|/|\orb(u)|$, and by vertex transitivity of $G$, we have $|\orb(u)| = |V(G)|$.
\end{proof}

\section{Lower bounds on the entangled capacity}
\label{sec:qcapacity}

In this section we lower bound the entangled capacity of  the graphs $G_n$ and~$H_n$. 
We start by dealing with the graph $G_n$. The graph $H_n$ will afterwards be treated  in a similar manner.

To prove the lower bounds we use a straightforward general method which was also used before in~\cite{Cubitt:2010, Leung:2012}.
Recall that a (real) $d$-dimensional orthonormal representation of a graph $G$ is a mapping $f:V(G)\to \R^d$ satisfying $f(u)\cdot f(u) = 1$ and $f(u)\cdot f(v)  = 0$ for every $\{u,v\}\in E(G)$.

\begin{proposition}\label{prop:leungbound}
If a graph $G$ has an orthonormal representation $f:V(G)\to \R^d$ and and has $M$ disjoint $d$-cliques, then $\Theta_q(G) \geq M$.
\end{proposition}

\begin{proof}
Let $\{1,\dots,M\}$ be a label set for the disjoint cliques. Let $\rho = I/d$ where $I$ is the $d$-by-$d$ identity matrix. 
For every $u\in V$ and $i\in\{1,\dots,M\}$ let $\rho(u)_i = f(u)f(u)^{\mathsf T}/d$ if $u$ belongs to the $i^{th}$ clique and let $\rho(u)_i$ be the zero matrix otherwise. Clearly these matrices are positive semidefinite and it is easy to check that they satisfy the conditions of Definition~\ref{def:qcap} using the fact that for every $d$-clique $C\subseteq V$, the set $\{f(u)\}_{u\in C}$ is a complete orthonormal basis for $\R^d$. This gives $\Theta_q(G)\geq \alpha_q(G)\geq M$.
\end{proof}


The lower bounds on the entangled capacity given in Lemma~\ref{lem:qcapacity} follow immediately from the following two lemmas and Proposition~\ref{prop:leungbound}.

\begin{lemma}\label{lem:orthrep}
Let $n$ be an odd integer. Then, the graph $G_n$ has an $n$-dimensional orthonormal representation.
\end{lemma}

\begin{lemma}\label{lem:cliques}
Let $n$ be such that there exists a Hadamard matrix of size~$n+1$.
Then, the graph~$G_n$ has at least $|V(G_n)|/n^2$ disjoint cliques of size~$n$.
\end{lemma}

We proceed by proving these lemmas.

\begin{proof}[ of Lemma~\ref{lem:orthrep}]
Associate with every vertex $x= (x_1,\dots,x_n)\in V$ a sign vector given by $u[x] = \big((-1)^{x_1},\dots, (-1)^{x_n}\big)^{\mathsf T}\in\R^n$. Let $\mathbf 1$ denote the $n$-dimensional all-ones vector. Note that for every $x\in V$, we have $u[x]\cdot \mathbf{1}= -1$, as the Hamming weight of $x$ is $(n+1)/2$. Moreover, for every $\{x,y\} \in E$ we have
$u[x]\cdot u[y] = -1$, which follows from the fact that $d(x,y) = (n+1)/2$.

Now consider the $(n+1)$-dimensional unit vectors $f(x) = (u[x]\oplus 1)/\sqrt{n+1}$ (\ie, the column vector $u[x]$ with a $1$ appended to it, normalized).
These vectors satisfy:
\begin{enumerate}
\item for every $\{x,y\}\in E$ we have
\beqn
f(x)\cdot f(y) = \frac{(u[x]\oplus 1)\cdot (u[y]\oplus 1)}{n+1} =0,
\eeqn
\item for every $x\in V$ we have
\beqn
f(x)\cdot (\mathbf 1\oplus 1) = \frac{(u[x]\oplus 1)}{\sqrt{n+1}}\cdot (\mathbf 1\oplus 1) =  \frac{-1 + 1}{\sqrt{n+1}} = 0.
\eeqn 
\end{enumerate}

The first item shows that $f$ forms an orthonormal representation of $G$.
The second item says that the vectors $\big(f(x)\big)_{x\in V}$ lie on a single $n$-dimensional hyperplane (orthogonal to the all-ones vector). Hence these vectors span a space of dimension at most~$n$.
It follows that there is an $n$-dimensional orthonormal representation of~$G_n$.
\end{proof}

To prove Lemma~\ref{lem:cliques} we need to find a large number of disjoint $n$-cliques in $G_n$. We achieve this by first finding just one $n$-clique. 
Using the fact that $G_n$ is vertex transitive, we show that existence of a single clique implies the existence of many disjoint cliques.
More explicitly, one can produce many pairwise disjoint $n$-cliques by simultaneously permuting the coordinates of the strings in this one clique. Notice that this permutation operation leaves both the Hamming weights and the Hamming distances invariant. A suitable choice of such permutations give pairwise disjoint cliques from any single clique, as whether or not a set of $n$-bit strings forms a clique in $G_n$ depends only on their Hamming weights and Hamming distances.


The following proposition tells us when we can find a single $n$-clique in~$G_n$.

\begin{proposition}\label{prop:GnHadamard}
Let $n$ be such that there exists a Hadamard matrix of size $(n+1)$. Then, there exists an $n$-clique in~$G_n$.
\end{proposition}

\begin{proof}
Let $M$ be an $(n+1)\!\times\! (n+1)$ Hadamard matrix. We may assume that the first row and column of $M$ contain only~$+1$s, since multiplying all entries in a row (or column) by $-1$, gives again a Hadamard matrix. Since each of the last~$n$ rows of $M$ is orthogonal to the first row, it has exactly $(n+1)/2$ entries equal to~$-1$.
Moreover, since each pair from the last~$n$ rows of $M$ is orthogonal, the two rows differ in exactly $(n+1)/2$ coordinates.

Let $C$ be the the $n\!\times\! n$ matrix obtained by removing the first row and column from~$M$.
Then, each row of $C$ has exactly $(n+1)/2$ entries equal to~$-1$, and every pair of rows from $C$ differ in exactly $(n+1)/2$ coordinates.
Hence, the rows of $C$ are a clique in~$G_n$.
\end{proof}

Next, we lower bound the number of disjoint $n$-cliques of size $n$ in $G_n$. 
We use the following lemma and proposition.

\begin{lemma}\label{lem:transcliques}
Let $G$ be a vertex transitive graph that has a $d$-clique as an induced subgraph. Then, $G$ has at least $|V(G)|/d^2$ vertex-disjoint induced $d$-cliques.
\end{lemma}

\begin{proof}
Let $W=C_1\cup C_2\cup\cdots\cup C_k$ be a union of $k$ disjoint $d$-cliques, with $k$ maximal. Since $G$ is vertex transitive, Corollary~\ref{cor:utov} implies that for every pair of vertices $u, v$ there are exactly $|\aut (G)|/|V(G)|$ automorphisms mapping $u$ to $v$. It follows that at most $$|W|\cdot |C_1|\cdot |\aut(G)|/|V(G)|$$ automorphisms map a vertex in $C_1$ to a vertex in $W$. 

On the other hand, by maximality of $k$, $\sigma(C_1)\cap W$ is nonempty for every automorphism $\sigma$. It follows that $|W|\cdot |C_1|\geq |V|$, and hence $k=|W|/|C_1|\geq |V|/|C_1|^2=|V|/d^2$.
%
\end{proof}

\begin{proposition}\label{prop:Gtrans}
For every $n$, the graph $G_n$ is vertex transitive.
\end{proposition}

\begin{proof}
Consider the group $S_n$ of permutations on $\{1,\dots,n\}$. For every $\sigma\in S_n$ define the map $\Gamma_\sigma:\bset{n}\to\bset{n}$ by $\Gamma_\sigma(x) = (x_{\sigma(1)},\dots,x_{\sigma(n)})$. As $\Gamma_\sigma$ leaves the Hamming weight invariant, we have $\Gamma_\sigma:V(G_n)\to V(G_n)$.
Moreover, $\Delta\big(\Gamma_\sigma(x),\Gamma_\sigma(y)\big) = \Delta(x,y)$. Hence, $\Gamma_\sigma\in\aut(G_n)$. Finally, for every $x\in V(G_n)$ we have $\{\Gamma_\sigma(x)\st \sigma\in S_n\} = V(G_n)$, and we are done.
\end{proof}

\begin{proof}[ of Lemma~\ref{lem:cliques}]
The result follows  by combining Propositions~\ref{prop:GnHadamard} and~\ref{prop:Gtrans} and Lemma~\ref{lem:transcliques}.
\end{proof}
\medskip

We deal with the graphs~$H_n$ in the same way as we did with the graphs~$G_n$.
We directly obtain the result of Lemma~\ref{lem:qcapacity} for these graphs by combining the following two lemmas with Proposition~\ref{prop:leungbound}.

\begin{lemma}\label{lem:orthrep-cube}
Let $n$ be an odd integer. Then, $H_n$ has an orthonormal representation of dimension $n+1$.
\end{lemma}

\begin{lemma}\label{lem:cliques-cube}
Let $n$ be such that there exists a Hadamard graph of size~$n$.
Then, the graph $H_n$ has at least $|V(H_n)|/(n+1)^2$ disjoint cliques of size~$n+1$.
\end{lemma}

\begin{proof}[ of lemma~\ref{lem:orthrep-cube}]
Associate with every vertex $x= (x_1,\dots,x_n)\in V$ the vector $$u[x] = \big((-1)^{x_1},\dots, (-1)^{x_n}\big)^{\mathsf T}.$$ Then, the unit vectors $f(x) = (u[x]\oplus 1)/\sqrt{n+1}$ form an $(n+1)$-dimensional orthonormal representation of $H_n$.
\end{proof}

To prove Lemma~\ref{lem:cliques-cube} we proceed as in the proof of Lemma~\ref{lem:cliques}: We first find a single $(n+1)$-clique in~$H_n$. Then we prove that~$H_n$ is vertex transitive and use Lemma~\ref{lem:transcliques}.

\begin{proposition}\label{prop:HnHadamard}
Let $n$ be such that there exists a Hadamard matrix of size~$(n+1)$. Then, there exists an $(n+1)$-clique in~$H_n$.
\end{proposition}

\begin{proof}
Let $C$ be an $n$-clique in the graph~$G_n$.
Then, since each of the vertices in $C$ has Hamming weight $(n+1)/2$, the union of $C$ and the all-zeroes string gives an $(n+1)$-clique in~$H_n$. The result now follows from Proposition~\ref{prop:GnHadamard}.
\end{proof}

\begin{proposition}\label{prop:Htrans}
For every $n$, the graph $H_n$ is vertex transitive.
\end{proposition}

\begin{proof}
Recall that $V(H_n)\subseteq \F_2^n$ consists of the strings of even Hamming weight.
For every $z\in V(H_n)$ define the linear bijection $\Sigma_z:\F_2^n\to\F_2^n$ by $\Sigma_z(x) = x+z$. As $\Sigma_z$ leaves the parity of Hamming weight invariant, we have  $\Sigma_z:(H_n)\to V(H_n)$.
Moreover, $\Delta\big(\Sigma_z(x),\Sigma_z(y)\big) = \Delta(x,y)$. Hence, $\Sigma_z\in\aut(G_n)$. For every $x\in V(H_n)$ we have $\{\Sigma_z(x)\st \sigma\in z\in V(H_n)\} = V(H_n)$, and we are done.
\end{proof}

\begin{proof}[ of Lemma~\ref{lem:cliques-cube}]
The result follows  by combining Propositions~\ref{prop:HnHadamard} and~\ref{prop:Htrans} and Lemma~\ref{lem:transcliques}.
\end{proof}

\section{Upper bounds on the Shannon capacity}
\label{sec:shannon}

In this section we upper bound  the Shannon capacity of the graphs $G_n$ and $H_n$. We recall that $G_n$ has as vertex set all  binary strings of odd length $n$ and Hamming weight~$(n+1)/2$, and as edge set the pairs of vertices with Hamming distance~$(n+1)/2$.
The graph $H_n$ has as vertex set all binary strings of odd length $n$ and even Hamming weight, and as edge set the pairs of vertices with Hamming distance~$(n+1)/2$. 
The proof of the upper bounds in Lemmas~\ref{lem:shannonbound}  is based on  a general method of Haemers~\cite{Haemers:1978} and an algebraic lemma of Frankl and Wilson~\cite{Frankl:1981}.

\begin{lemma}[Haemers~\cite{Haemers:1978}]\label{lem:haemers}
Let $G = (V,E)$ be a graph. Let $F$ be a field. Let $A:V\times V\to F$ be a matrix such that for every $x\in V$ we have $A(x,x)\ne 0$ and for every non-adjacent pair $x,y\in V$ we have $A(x,y) = 0$. Then, $\Theta(G) \leq \rank(A)$.
\end{lemma}

\begin{proof}
Say that a matrix $A:V\times V\to F$ {\em fits $G$} if it satisfies the conditions stated in the lemma. Let $S\subseteq V$ be a maximum-sized independent set and let $A$ be a matrix that fits $G$. Then, the principle submatrix of $A$ defined by $S$ has rank $|S|$. Hence, we have $\alpha(G) \leq \rank(A)$. The result follows because $A^{\otimes n}$ fits $G^{\boxtimes n}$ and $\rank(A^{\otimes n}) = \rank(A)^n$.
\end{proof}

We say that a polynomial is {\em multilinear} if its degree in each variable is at most~1.

\begin{lemma}[Frankl-Wilson~\cite{Frankl:1981}]\label{lem:kalai}
Let $p$ be an odd prime, let $r$ be a natural number and let $n = rp-1$. Let $\mathcal V\subseteq \pmset{n}\subseteq \F_p^n$ be a set of vectors over $\F_p$. Then, for every $u\in\mathcal V$ there exists a multilinear polynomial $P_u\in\F_p[v_1,\dots,v_n]$ satisfying:
\begin{enumerate}
\item $P_u(u) \ne 0$,\\[.1cm]
\item for every $v\in\mathcal V$ such that $\langle u,v\rangle\ne -1$, we have $P_u(v) =0$,\\[.1cm]
\item $\deg(P_u) \leq p-1$.
\end{enumerate}
\end{lemma}

\begin{proof}
For every vector $u\in \mathcal V$ let $Q_u\in\F_p[v_1,\dots,v_n]$ be the polynomial defined  by
\beqn
Q_u(v) = \prod_{i=1}^{p-1}\big(\langle u,v\rangle + 1 - i\big),
\eeqn

Since $n\equiv -1\pmod p$, every $v\in \mathcal V$ satisfies $\langle v,v\rangle = -1$. 
By Wilson's Theorem (see for example Lidl and Niederreiter~\cite{Lidl:1983}) it follows that $Q_u(u) = (-1)^{p-1}(p-1)! =  (-1)^p$.
If $\langle u,v\rangle \ne -1$ we have $Q_u(v) = 0$ since in this case we have $\langle u,v\rangle +1 \in\{1,\dots,p-1\}$.
In particular, we have $Q_u(u)\ne 0$.
Since $v\mapsto \langle u, v\rangle$ is a linear function, we have  $\deg(Q_u) = p-1$.

Define the multilinear polynomial $P_u$ by expanding $Q_u$ in the monomial basis and changing the powers $t_i$ the monomial $v_1^{t_1}\cdots v_n^{t_n}$ to 0 if $t_i$ is even and to 1 if $t_i$ is odd. Then $P_u$ is multilinear and agrees with $Q_u$ everywhere on $\pmset{n}$ and satisfies $\deg(P_u) \leq \deg(Q_u)$.
\end{proof}

We now show how these two lemmas can be combined to give Lemma~\ref{lem:shannonbound}, which states that for~$p$ an odd prime and $n = 4p-1$, and $G$ either $G_n$ or $H_n$, we have $\Theta(G) \leq O\big(2^{0.812 n}\big)$.

\begin{proof}[ of Lemma~\ref{lem:shannonbound}]
Let $(V,E)$ be either $G_n$ or $H_n$.
For every $x\in V$ let $u[x] = \big((-1)^{x_1},\dots,(-1)^{x_n}\big)$ be the corresponding sign vector in~$\F_p^n$. 
Since $n = 4p-1\equiv -1\pmod p$ and $\F_p$ is isomorphic to the ring of integers mod $p$, we have for every $x,y\in V$,
\beq\label{eq:ipdist}
\langle u[x],u[y]\rangle = n - 2d(x,y) = -2d(x,y) - 1.
\eeq

Let $x,y\in V$ be distinct vertices such that $\{x,y\}\not\in E$ is a non-edge.
We claim that $\langle u[x],u[y]\rangle\ne -1$.
It is not hard to see that if two strings have even Hamming weight, then their Hamming distance is also even.
Hence, we have $d(x,y) \in\{2,4,\dots,4p-2\}$. 
Moreover, since $p$ is odd, the only possible multiple of $p$ that $2d(x,y)$ can attain is $4p$.
Since edges are formed by pairs with Hamming distance~$2p$, we have $2d(x,y)\ne 4p$.
This implies that $2d(x,y)\not\equiv 0\pmod p$, and the claim follows form Eq.~\eqref{eq:ipdist}.

Set $r = 4$ and let $\mathcal V = \{u[x]\st x\in V\}$. Then, Lemma~\ref{lem:kalai} gives a multilinear polynomial $P_x\in\F_p[v_1\dots,v_n]$ for every $x\in V$, satisfying:
\begin{enumerate}
\item $P_x(u[x]) \ne 0$,\\[.1cm]
\item for every $y\in V$ such that $y\ne x$ and $\{x,y\}\not\in E$, we have $P_x(u[y]) = 0$,\\[.1cm]
\item $\deg(P_x) \leq p-1$.
\end{enumerate}

The set $\mathcal{M}$ of multilinear monomials in $n$ variables of degree at most $p-1$ forms a basis for the space of multilinear polynomials of degree at most $p-1$. For every vertex $x\in V$, define vectors $S[x],T[x]\in\F_p^{\mathcal{M}}$ as follows. For monomial $m\in \mathcal{M}$ let $S[x]_m$ be the coefficient of $m$ in the expansion of the polynomial $P_x$ in the basis $\mathcal{M}$ and let $T[x]_m= m\big(u[x]\big)$ be the value obtained by evaluation the monomial $m$ at $u[x]$. Then, for every $x,y\in V$ we have $S[x]\cdot T[y] = P_x(u[y])$.

Consider now the matrix $A:V\times V\to \F_p$ defined by $A(x,y) = S[x]\cdot T[y]$. Since the vectors $S[x]$ and $T[y]$ have dimension $|\mathcal M|$, we have $\rank(A) \leq |\mathcal M|$. Additionally, it follows from the properties of the polynomials $P_x$ that the matrix $A$ satisfies $A(x,x)\ne 0$ for every $x\in V$ and $A(x,y) = 0$ for every non-adjacent pair $x,y\in V$.

The claim now follows from Lemma~\ref{lem:haemers} and the fact that
\beqn
|\mathcal M| = \sum_{i=0}^{p-1}{n\choose i}.
\eeqn
This completes the proof.
\end{proof}

\section*{Acknowledgements}
JB thanks Oded Regev for stimulating discussions during a pleasant visit to ENS Paris, and helpful comments on an earlier version of this manuscript. We thank Monique Laurent and David Garc\'ia~Soriano for useful discussions and Lex Schrijver for useful pointers to the literature.


\begin{thebibliography}{CLMW10}
\expandafter\ifx\csname urlstyle\endcsname\relax
  \providecommand{\doi}[1]{doi:\discretionary{}{}{}#1}\else
  \providecommand{\doi}{doi:\discretionary{}{}{}\begingroup
  \urlstyle{rm}\Url}\fi

\bibitem[AB95]{Alperin:1995}
J.~Alperin and R.~B. Bell.
\newblock \emph{{Groups and Representations}}.
\newblock Number 162 in Graduate texts in mathematics. Springer, New York, NY,
  1995.

\bibitem[AGR81]{Aspect:1981}
A.~Aspect, P.~Grangier, and G.~Roger.
\newblock Experimental tests of realistic local theories via \uppercase{B}ell's
  theorem.
\newblock \emph{Phys. Rev. Lett.}, 47(7):460--463, 1981.
\newblock \doi{10.1103/PhysRevLett.47.460}.

\bibitem[AHKS06]{Avis:2006}
D.~Avis, J.~Hasegawa, Y.~Kikuchi, and Y.~Sasaki.
\newblock {A quantum protocol to win the graph colouring game on all Hadamard
  graphs}.
\newblock \emph{IEICE transactions on fundamentals of electronics,
  communications and computer science}, 89(5):1378--1381, 2006.

\bibitem[BCT99]{Brassard:1999}
G.~Brassard, R.~Cleve, and A.~Tapp.
\newblock Cost of exactly simulating quantum entanglement with classical
  communication.
\newblock \emph{Phys. Rev. Lett.}, 83:1874--1877, 1999.

\bibitem[BCW98]{Buhrman:1998}
H.~Buhrman, R.~Cleve, and A.~Wigderson.
\newblock Quantum vs. classical communication and computation.
\newblock In \emph{Proceedings of the 30th Annual ACM Symposium on Theory of
  Computing (STOC 1998)}, pages 63--68. 1998.

\bibitem[Bel64]{Bell:1964}
J.~S. Bell.
\newblock On the \uppercase{E}instein-\uppercase{P}odolsky-\uppercase{R}osen
  paradox.
\newblock \emph{Physics}, 1:195--200, 1964.

\bibitem[BSST02]{Bennett:2002}
C.~Bennett, P.~Shor, J.~Smolin, and A.~Thapliyal.
\newblock Entanglement-assisted capacity of a quantum channel and the reverse
  shannon theorem.
\newblock \emph{Information Theory, IEEE Transactions on}, 48(10):2637--2655,
  2002.

\bibitem[CLMW10]{Cubitt:2010}
T.~S. Cubitt, D.~Leung, W.~Matthews, and A.~Winter.
\newblock Improving zero-error classical communication with entanglement.
\newblock \emph{Phys. Rev. Lett.}, 104(23):230503, 2010.

\bibitem[CMN{\etalchar{+}}07]{Cameron:2007}
P.~Cameron, A.~Montanaro, M.~Newman, S.~Severini, and A.~Winter.
\newblock On the quantum chromatic number of a graph.
\newblock \emph{The electronic journal of combinatorics}, 14(R81):1, 2007.

\bibitem[EPR35]{Einstein:1935}
A.~Einstein, P.~Podolsky, and N.~Rosen.
\newblock Can quantum-mechanical description of physical reality be considered
  complete?
\newblock \emph{Physical Review}, 47:777--780, 1935.

\bibitem[FR87]{Frankl:1987}
P.~Frankl and V.~R\"{o}dl.
\newblock Forbidden intersections.
\newblock \emph{Transactions of the American Mathematical Society},
  300(1):259--286, 1987.

\bibitem[FW81]{Frankl:1981}
P.~Frankl and R.~Wilson.
\newblock Intersection theorems with geometric consequences.
\newblock \emph{Combinatorica}, 1(4):357--368, 1981.

\bibitem[GN08]{Godsil:2008}
C.~Godsil and M.~Newman.
\newblock {Coloring an Orthogonality Graph}.
\newblock \emph{SIAM Journal on Discrete Mathematics}, 22:683, 2008.

\bibitem[Hae78]{Haemers:1978}
W.~Haemers.
\newblock An upper bound for the shannon capacity of a graph.
\newblock In \emph{Colloq. Math. Soc. J{\'a}nos Bolyai}, volume~25, pages
  267--272. 1978.

\bibitem[LMM{\etalchar{+}}12]{Leung:2012}
D.~Leung, L.~Mancinska, W.~Matthews, M.~Ozols, and A.~Roy.
\newblock Entanglement can increase asymptotic rates of zero-error classical
  communication over classical channels.
\newblock \emph{Communications in Mathematical Physics}, 311:97--111, 2012.
\newblock ISSN 0010-3616.
\newblock \doi{10.1007/s00220-012-1451-x}.

\bibitem[LN83]{Lidl:1983}
R.~Lidl and H.~Niederreiter.
\newblock Finite fields.
\newblock In G.-C. Rota, editor, \emph{Finite Fields}, volume~20 of
  \emph{Encyclopedia of Mathematics and its Applications}. Addison-Wesley,
  Reading, Massachusetts, 1983.

\bibitem[MSS12]{Mancinska:2012}
L.~Man\v{c}inska, G.~Scarpa, and S.~Severini.
\newblock {A Generalization of Kochen-Specker Sets Relates Quantum Coloring to
  Entanglement-Assisted Channel Capacity}.
\newblock 2012.
\newblock Available at arXiv: 1207.1111 [quant-ph].

\bibitem[NC00]{Nielsen:2000}
M.~A. Nielsen and I.~L. Chuang.
\newblock \emph{Quantum Computation and Quantum Information}.
\newblock Cambridge University Press, New York, 2000.

\bibitem[Pal33]{Paley:1933}
R.~E.~A.~C. Paley.
\newblock {On orthogonal matrices}.
\newblock \emph{Journal of Mathematics and Physics (now called Studies in
  Applied Mathematics)}, 12:311--320, 1933.

\bibitem[Sca98]{Scarpis:1898}
U.~Scarpis.
\newblock Sui determinanti di valore massimo.
\newblock \emph{Rendiconti Reale Instituto Lombardo di Scienze e Lettere (Milan
  Rendeiconti)}, 31:1441--1446, 1898.

\bibitem[Sha56]{Shannon:1956}
C.~Shannon.
\newblock The zero error capacity of a noisy channel.
\newblock \emph{Information Theory, IRE Transactions on}, 2(3):8--19, 1956.

\end{thebibliography}

\newcommand{\etalchar}[1]{$^{#1}$}

\end{document}